\documentclass[a4paper, 12pt]{article}
\usepackage{amssymb,amsmath,latexsym,color,graphicx}
\usepackage{enumerate}
\usepackage{color}

\def\wrtext#1{\relax\ifmmode{\leavevmode\hbox{#1}}\else{#1}\fi}
\def\abs#1{\left|#1\right|}
\def\begeq{\begin{equation}}
\def\endeq{\end{equation}}

\def\part#1{\frac{\partial}{\partial #1}}

\newcommand{\real}{\mbox{\bf R}}
\newcommand{\comp}{\mbox{\bf C}}

\newcommand{\nat}{\mbox{\bf N}}

\renewcommand{\Re}{\mbox{\rm Re\,}}

\renewcommand{\exp}{\mbox{\rm exp\,}}

\def\bdy{boundary}

\def\e{equation}
\def\fu{function}

\def\neigh{neighborhood}

\def\Re{{\rm Re\,}}

\newtheorem{dref}{Definition}[section] 
\newtheorem{theo}[dref]{Theorem} \newtheorem{prop}[dref]{Proposition}
\newtheorem{remark}[dref]{Remark} 
\newtheorem{cor}[dref]{Corollary}

\newenvironment{proof}{\par\noindent{{\bf Proof.}}}{\hfill$\Box$
\medskip}

\newcommand{\ekv}[2]{\begin{equation}\label{#1}#2\end{equation}}

\catcode`\?=\active\def ?{\'e} \catcode`\?=\active\def ?{\`e}
\catcode`\?=\active\def ?{\c{c}} \catcode`\?=\active\def ?{\`a}
\catcode`\?=\active\def ?{\^e} \catcode`\?=\active\def ?{\`u}
\catcode`\?=\active\def ?{\^o} \catcode`\?=\active\def ?{\^u}
\catcode`\?=\active\def ?{\^{\i}} \catcode`\?=\active\def ?{\^a}
\catcode`\?=\active\def ?{\"o}

\def\mmm{{\cal M}}

\def\R{\mathbb R}

\def\D{\partial}

\def\abs#1{\left\vert#1\right\vert}
\def\set#1{\left\{#1\right\}}
\def\sep#1{\left(#1\right)}
\def\Re{{\mathrm Re\,}}

\newcommand{\preuve}[1][\!\!]{\noindent{\it Proof #1. \ \ }}

\renewcommand{\R}{\mbox{\bf R}}

\author{ Fr\'ed\'eric H\'erau\footnote{{This work was supported by the French Agence Nationale de la Recherche, NOSEVOL project,  ANR 2011 BS01019 01}}
\\\small Laboratoire de Math\'ematiques Jean Leray\\\small Universit\'e de Nantes\\\small 2, rue de la Houssini\`ere \\\small BP 92208, 44322 Nantes Cedex 3,
France\\\small and FRE 3111 CNRS\\\small herau@univ-reims.fr \and Michael Hitrik \\\small Department of Mathematics \\\small University of California
  \\\small Los Angeles \\\small CA 90095-1555, USA\\\small hitrik@math.ucla.edu \and
  Johannes Sj\"ostrand$^\dagger$
  \\\small IMB, Universit\'e de Bourgogne\\
  \small 9, Av. A. Savary, BP 47870\\
  \small FR-21078 Dijon C\'edex\\
  \small  and UMR 5584, CNRS\\
  \small johannes.sjostrand@u-bourgogne.fr} \date{} \title{Tunnel effect and
  symmetries for Kramers Fokker-Planck type operators\\
}

\title{Supersymmetric structures for second order differential operators \footnote{In
    memory of Vladimir Buslaev}}

\begin{document}

\maketitle

\vspace*{1cm}
\noindent
{\bf Abstract}: Necessary and sufficient conditions are obtained for a real semiclassical partial differential operator of order two to
possess a supersymmetric structure. For the operator coming from a chain of oscillators, coupled to two heat baths, we show the non-existence of a smooth
supersymmetric structure, for a suitable interaction potential, provided that the temperatures of the baths are different.

%

\tableofcontents
\section{Introduction and statement of results}
\setcounter{equation}{0}
In a large number of problems coming from statistical or quantum mechanics, involving real partial differential operators of order two with the spectrum contained
in the right half-plane, one is often interested in the splitting between the two smallest real parts of the eigenvalues, at least when it is possible to show
that the first one is simple and isolated. This kind of issue may concern the Schr\"odinger operator or the Witten Laplacian in quantum mechanics, and the
Kramers-Fokker-Planck operator, or more generally, some models of chains of oscillators coupled to heat baths, in statistical mechanics.

\medskip
\noindent
In the semiclassical context, the tools available for studying the splitting are very powerful, and allow for a detailed analysis for models where a natural small
parameter is in play, e.g. the Planck constant or the low temperature parameter. In some cases the splitting may be exponentially small with respect to
the parameter, and a so-called tunneling effect may appear, related to a very low rate of convergence for the associated evolution problem. It is often related
to some degenerate geometry in the model, such as the presence of multiple confining wells.

\medskip
\noindent
In fact, for some of the preceding semiclassical differential operators, a particularly convenient structure is sometimes available,
which simplifies considerably the analysis of the splitting and of the tunneling effect. Indeed, when the operators have a Hodge Laplacian type structure of the
form $P = d^*d$, where $d$ is a (perhaps modified)  de Rham operator, the eigenspaces have some natural stability properties with respect to the operator $d$, and
the study of the eigenvalues can be reduced to the study of the singular values of $d$ on finite dimensional subspaces. This type of operators is sometimes called
\it supersymmetric \rm (see in particular \cite{Wit82} and the full definition given below), and this method was successfully employed for determining the splitting
corresponding to some models mentioned above.

\medskip
\noindent
The natural question raised in this article is to determine under which conditions a given (not necessarily selfadjoint) operator is of supersymmetric type.
The answer given here implies that for some models describing systems out of equilibrium, there is no such property, and we give a concrete example
in the last section of the paper. The significance of this kind of result is that the study of the splitting for these models may be far more complicated,
in particular in the non-selfadjoint case, and to the best of our knowledge, no such analysis has been performed in the non-selfadjoint case so far.

\bigskip
\noindent
By way of introduction and also, for motivation, let us mention now some examples of supersymmetric differential operators and related results about the eigenvalue
splitting. The first example is the Witten-Hodge Laplacian,
\begin{equation} \label{witten}
W = -h^2 \Delta + (\D_x V(x))^2 - h \Delta V(x),
\end{equation}
on ${\real}^n_x$. Here $V\in C^{\infty}(\real^n;\real)$ is a Morse function with $\partial^{\alpha} V \in L^{\infty}(\real^n)$ for $\abs{\alpha}\geq 2$, and such
that $| \D_x V(x) | \geq 1/C$ when $|x| \geq C>0$. The second example is the Kramers-Fokker-Planck operator given by
\begin{equation}
\label{KFP}
K = -h^2 \Delta_y + y^2 - hn + y\cdot h\D_x - \D_x V(x) \cdot h \D_y,
\end{equation}
on ${\real}^{2n}_{x,y}$, where the potential $V$ has the same properties as above. This is a kinetic (non-selfadjoint) model for an oscillator
coupled to a heat bath.

\medskip
\noindent
For these two models, the splitting can be evaluated very precisely in the semiclassical limit $h\rightarrow 0$. Considering $W$ and $K$ as unbounded
operators on $L^2$, it was shown in \cite{HeSj85}, \cite{HKN08} for the Witten Laplacian, and in \cite{HHS08a}, \cite{HHS08b} for the Kramers-Fokker-Planck
operator, that, under suitable assumptions, it is exponentially small. We notice here that in both cases, the operators have non-negative symmetric parts and
have $\mu_1= 0$ as a simple eigenvalue, under the additional assumption that $V(x)\rightarrow \infty$ as $x\rightarrow \infty$. For simplicity, let us recall here
the result in the case when the potential $V$ has two local minima $m_\pm$ and a saddle point $s_0$ of signature $(n-1, 1)$, although a far more general result
is given in~\cite{HHS11}, in the Kramers-Fokker-Planck case. Let $S_\pm = V(s_0) - V(m_\pm)$ and $S = \min(S_+, S_-)$. Then there exists $c>0$ and  $h_0 >0$
such that for all $0<h\leq h_0$ there are precisely two eigenvalues of $K$ in the open disc $D(0,ch)$, one of them being $\mu_1=0$, and the second one is real
positive, given by
\begin{equation} \label{mu2}
\mu_2 = h l(h) e^{-S/h},
\end{equation}
where $l(h) = l_0 + h l_1 + h^2 l_2+ h^3 l_3... $, with $l_0>0$.

\bigskip
\noindent
We shall now discuss the notion of a supersymmetric structure more precisely. Let $X$ be $\real^n$ or a smooth compact manifold of dimension $n$, equipped with a
smooth strictly positive volume density $\omega(dx)$. We can then choose local coordinates $x_1,..,x_n$ near any given point on $X$ such that
$\omega (dx)=dx_1\cdot ...\cdot dx_n$.

\medskip
\noindent
Let $d:C^\infty (X;\wedge ^kT^*X)\to C^\infty (X;\wedge ^{k+1}T^*X)$ be the de Rham complex, and let
$\delta :C^\infty (X;\wedge^{k+1}TX)\to C^\infty (X;\wedge ^kTX)$ be the adjoint of $d$ with respect to the natural pointwise duality between
$\wedge^kT^*X$ and $\wedge^kTX$, integrated against $\omega$. In the special local coordinates above, we have
\begin{equation}
\label{ddelta}
d=\sum_{j=1}^n dx_j^\wedge\circ\frac{\partial }{\partial x_j},\quad \delta =-\sum_{j=1}^n \frac{\partial }{\partial x_j}\circ dx_j^\rfloor
\end{equation}

\medskip
\noindent
Let $A(x):T_x^*X\to T_xX$ be a linear map, depending smoothly on $x\in X$. The $k$--fold exterior product $\wedge^kA$ maps $\wedge^kT_x^*X\to
\wedge^kT_xX$, and by abuse of notation, we shall sometimes write simply $A$ instead of $\wedge^kA$. By convention, $\wedge^0 A$ is the identity map on
${\bf R}$. Associated with $A$, we have the bilinear product on the space of compactly supported smooth $k$--forms,
\begin{equation}
\label{bil}
(u|v)_A=(Au|v)= \int (A(x)u(x)|v(x))\omega (dx),\ u,v\in C_0^\infty (X;\wedge^kT_x^*X).
\end{equation}
Since we shall restrict the attention to real-valued sections and operators, there is no need for complex conjugations in (\ref{bil}). When the map $A(x)$ is
bijective for each $x$, there is a natural way of defining the formal adjoint of a linear operator $B$ taking $k$--forms to $\ell$--forms on $X$. This adjoint,
denoted by $B^{A,*}$, maps $\ell$--forms to $k$--forms and is given by
$$
(Bu|v)_A=(u|B^{A,*}v)_A,
$$
and more explicitly, by the following expression,
$$
B^{A,*}=(\wedge^k A^{\mathrm{t}})^{-1} B^* \wedge^{\ell} A^{\mathrm{t}}.
$$
Here $B^*$ is the usual adjoint taking $\ell$--vectors to $k$--vectors. Notice that when $k=0$, we can still define the adjoint $B^{A,*}$,
even when $A$ is not everywhere invertible. In the case when $B=d$, the de Rham differentiation, we get
$$
d^{A,*}=(A^{\mathrm{t}})^{-1} \delta A^{\mathrm{t}},
$$
and for the adjoint of the restriction of $d$ to $0$--forms, we get
$$
d^{A,*}=\delta A^\mathrm{t}.
$$
We can write, in the special local coordinates above,
$$
d^{A,*} = -\sum_{j=1}^n \frac{\partial }{\partial x_j}\circ dx_j^\rfloor \circ A^\mathrm{t}.
$$

\bigskip
\noindent
Let $\varphi\in C^{\infty}(X;\real)$ and let us introduce the Witten complex, given by the weighted semiclassical de Rham differentiation
\begin{equation}
\label{dphi}
d_{\varphi,h} =e^{-\varphi/h} \circ hd \circ e^{\varphi/h} = hd + (d\varphi)^\wedge: C^{\infty}_0(X; \wedge^k T^*X) \rightarrow C^{\infty}_0(X; \wedge^{k+1} T^*X).
\end{equation}
Considering the adjoint of $d_{\varphi,h}$ with respect to the bilinear product (\ref{bil}), we can introduce the operator,
\begin{equation}
\label{dphiAstar}
d_{\varphi,h}^{A,*} = e^{\varphi/h}\circ hd^{A,*} \circ e^{-\varphi/h},
\end{equation}
which on $1$--forms becomes,
$$
d_{\varphi,h}^{A,*} = (h\delta + d\varphi^\rfloor)\circ A^\mathrm{t}.
$$

\bigskip
\noindent
We are now able to give a precise definition of the notion of a supersymmetric differential operator. In doing so, it will be convenient to distinguish between the
non-semiclassical case, corresponding to the situation where the parameter $h>0$ is kept fixed, say, $h=1$, and the semiclassical case, where we shall let
$h\in (0,h_0]$, $h_0>0$.

\begin{dref}
{\rm (i)} (The non-semiclassical case.) Let $P=P(x,D_x)$ be a second order scalar real differential operator on $X$, with $C^{\infty}$--coefficients.
We say that $P$ has a supersymmetric structure on $X$ if there exists a linear map $A(x): T^*_x X \rightarrow T_x X$, smooth in $x\in X$, and functions
$\varphi,\psi \in C^{\infty}(X;\real)$ such that
$$
P = d_{\psi}^{A,*}d_{\varphi}.
$$
Here $d_{\varphi} = d_{\varphi,1}$, and similarly for $d_{\psi}$. \\
{\rm (ii)} (The semiclassical case.) Let $P = P(x,hD_x;h)$ be a second order scalar real semiclassical differential operator on $X$. We say that $P$
has a supersymmetric structure on $X$ in the semiclassical sense if there exists a linear $h$--dependent map $A(x;h): T^*_x X \rightarrow T_x X$,
smooth in $x\in X$, and functions $\varphi=\varphi(x;h)$ and $\psi=\psi(x;h)$, smooth in $x$, such that $\varphi=\varphi_0(x) + {\cal O}(h)$,
$\psi=\psi_0(x) + {\cal O}(h)$ in the $C^{\infty}$--sense, and for which we have
\begeq
\label{eq_super_sem}
P = d_{\psi,h}^{A,*}d_{\varphi,h},
\endeq
for all $h\in (0,h_0]$, $h_0>0$.
\end{dref}

\medskip
\noindent
{\it Remark}. Let us notice that no control on the behavior of $A(x;h)$ is required, as $h\rightarrow 0$, in the semiclassical case in Definition 1.1.
It is an interesting problem to determine when we can find $A(x;h)$ with some uniform control as $h\rightarrow 0$. As we shall see from the proof
of Theorem \ref{supsym}, the symmetric part of $A$ is immediate to determine, while the antisymmetric part arises as a solution of a $\delta$--problem with
exponential weights.

\bigskip
\noindent
Assuming that the linear map $A(x;h)$ is invertible for each $x\in X$ and all $h\in (0,h_0]$, let us put in an arbitrary degree $k\geq 0$,
\ekv{all.13bis}
{\Box_{A,\psi,\varphi}^{(k)} = d_{\psi,h}^{A,*}d_{\varphi,h} + d_{\varphi,h}d_{\psi,h}^{A,*}.}
The supersymmetric semiclassical operator $P$ in (\ref{eq_super_sem}) agrees with the restriction of $\Box_{A,\psi,\varphi}$ to $0$--forms.
The analysis of the eigenvalue splitting of $P$ strongly depends on the following formal intertwining relations,
\ekv{all.14}{
\Box_{A,\psi ,\varphi}^{(k+1)} d_{\varphi,h} = d_{\varphi,h} \Box_{A,\psi ,\varphi}^{(k)},\quad
d_{\psi,h}^{A,*}\Box_{A,\psi ,\varphi}^{(k+1)}=\Box_{A,\psi ,\varphi}^{(k)}d_{\psi,h}^{A,*}.}

\bigskip
\noindent
As a clarifying example, let us now come back to the case of the Witten Laplacian $W$ in (\ref{witten}). A straightforward computation shows that the operator
$W$ enjoys a semiclassical supersymmetric structure in the sense of Definition 1.1,
$$
W= (hd +  (d V)^\wedge)^{{\mathrm I},*} (hd + (dV)^\wedge) = d_{V,h}^{{\mathrm I},*} d_{V,h},
$$
with $\varphi = \psi = V$, and $A= {\mathrm I}$. In the case of the Kramers-Fokker-Planck operator $K$, the semiclassical supersymmetric
structure was observed in~\cite{Bi05},~\cite{BL08}, and in~\cite{TaTaKu06}. It will be recalled in the next section, following~\cite{HHS08a}.

\bigskip
\noindent
The first result of this paper is a necessary and sufficient condition for a second order real differential operator to be supersymmetric, both in the
non-semiclassical and in the semiclassical sense, either locally or globally on $X$.

\begin{theo}
\label{supsym}
{\rm (i)} (The non-semiclassical case). Let $P$ be a second order scalar real differential operator on $X$, which can be written
in local coordinates as follows,
$$
P=-\sum_{j,k=1}^n \partial_{x_j}\circ B_{j,k}(x) \circ \partial_{x_k} + \sum_{j=1}^n v_j(x)\circ \partial_{x_j}+v_0,
$$
where $(B_{j,k})$ is symmetric, and $(B_{j,k})$, $v_j$, $v_0$ are real-valued and smooth. In order for $P$ to have a supersymmetric structure on $X$
it is necessary that there should exist $\varphi,\psi \in C^{\infty}(X;\real)$ such that
\begeq
\label{all.11bis}
P(e^{-\varphi})=0,\quad P^*(e^{-\psi})=0.
\endeq
Here $P^*$ is the formal adjoint of $P$ with respect to the $L^2$--scalar product determined by the density of integration $\omega(dx)$. If
{\rm (\ref{all.11bis})} holds and the $\delta$--complex is exact in degree $1$ for smooth sections, then $P$ has a supersymmetric structure in the
sense of Definition {\rm 1.1}. \\
{\rm (ii)} (The semiclassical case.) Let $P$ be a second order scalar real semiclassical differential operator on $X$, which can be written
in local coordinates as follows,
\begeq
\label{thm1.2}
P=-\sum_{j,k=1}^n h\partial _{x_j} \circ B_{j,k}(x;h) \circ h\partial _{x_k} + \sum_{j=1}^n v_j(x;h)\circ h\partial _{x_j}+ v_0(x;h),
\endeq
where $(B_{j,k})$ is symmetric, and $(B_{j,k})$, $v_j$, $v_0$ are real-valued, $h$--dependent, smooth in $x\in X$. In order for $P$ to have a
supersymmetric structure on $X$ in the semiclassical sense it is necessary that there should exist functions $\varphi(x;h),\psi(x;h)$ smooth in $x\in X$, with
$\varphi=\varphi_0(x) + {\cal O}(h)$, $\psi = \psi_0(x) + {\cal O}(h)$ in the $C^{\infty}$--sense, such that
\begeq
\label{semicl1}
P(e^{-\varphi/h})=0,\quad P^*(e^{-\psi/h})=0,
\endeq
for all $h\in (0,h_0]$, $h_0>0$. If {\rm (\ref{semicl1})} holds and the $\delta$--complex is exact in degree $1$ for smooth sections, then $P$ has a
semiclassical supersymmetric structure in the sense of Definition {\rm 1.1}.
\end{theo}

\medskip
\noindent
{\it Remark}. Assume that $X$ is a compact Riemannian manifold and that $P = -\Delta + v$, where $v$ is a smooth real vector field. Then
{\rm (\ref{all.11bis})} holds, with $\varphi=\psi=0$, precisely when the vector field $v$ is divergence free.

\bigskip
\noindent
As a fundamental example of a second order semiclassical operator for which the supersymmetric structure may break down, we shall now discuss
the case of an operator associated to a model of a chain of anharmonic oscillators coupled to two heat baths. Referring to the next section
for the construction of this model, let us consider here the following real semiclassical differential operator of order two,
\begin{multline}
\label{eq1intro}
\widetilde{P}_W = \frac{\gamma}{2}\sum_{j=1}^2 \alpha_j(-h\partial _{z_j})\left(h\partial_{z_j}+\frac{2}{\alpha_j}(z_j-x_j)\right) \\
 + y\cdot h\partial_x - (\partial_x W(x)+x-z)\cdot h\partial_y.
\end{multline}
Here $x=(x_1,x_2)$, $y=(y_1,y_2)$, and $z=(z_1,z_2)$ belong to $\real^{2n}$. Furthermore, $\gamma>0$, $\alpha_j>0$, $j=1,2$. The effective
potential $W$ is of the form $W(x) = W_1(x_1) + W_2(x_2) + \delta W(x_1,x_2)$, where $W_j \in C^{\infty}(\real^n;\real)$, $j=1,2$, and
$\delta W\in C^{\infty}(\real^{2n};\real)$. The parameters $\alpha_j$ are proportional to the temperatures in the baths, $j=1,2$. As we shall see,
the structure of operator $\widetilde{P}_W$ is completely different depending on whether $\alpha_1 = \alpha_2$ or not. The following
is the second main result of this work --- see also Theorem \ref{thm_main} below for a more precise statement.

\begin{theo}
\label{nosupchains}
Consider the operator $\widetilde{P}_W$ defined in {\rm (\ref{eq1intro})}. We have the following two cases.
\begin{enumerate}[i)]
\item If $\alpha_1=\alpha_2$ (the equilibrium case) or $\delta W \equiv 0$ (the decoupled case), then $\widetilde{P}_W$ has a semiclassical
supersymmetric structure on $\R^{6n}$, in the sense of Definition {\rm 1.1}.
\item Take $\gamma=1$ and assume that $\alpha_1 \neq \alpha_2$. Then there exist a Morse function $W_1(x_1)$ with two local minima and a saddle point, a positive definite quadratic
form $W_2(x_2)$, and $\delta W(x_1,x_2) \in C_0^\infty(\R^{2n})$, arbitrarily small in the uniform norm, such that the operator
$\widetilde{P}_W$ has no supersymmetric structure in the semiclassical sense, on $\R^{6n}$, in the sense of Definition {\rm 1.1}.
\end{enumerate}
\end{theo}

\bigskip
\noindent
The study of the tunneling effect for the Witten Laplacian through its supersymmetric properties was first performed in \cite{HeSj85},
following the supersymmetry observation pointed out in \cite{Wit82}. As already mentioned, this structure is of great help in the study of the eigenvalue
splitting, since in particular it allows one to avoid the study near the so--called "non-resonant wells", such as the saddle point in the example mentioned above.
The results were subsequently generalized in \cite{BEGK},\cite{BGK1}, in \cite{HKN08} with a full asymptotic expansion, in
\cite{Lep10a}, \cite{Lep10b} in cases with boundary, and in \cite{Lep11}, \cite{LNV12} in the case of forms of higher degree. Notice
that the computation of the exponentially small eigenvalues is performed here using the singular values of the Witten differential.

\medskip
\noindent
For the Kramers-Fokker-Planck operator, the supersymmetric structure was first observed in \cite{Bi05}, \cite{BL08}, and \cite{TaTaKu06}, and was used for a
complete study of the tunneling effect in \cite{HHS08a} and \cite{HHS11}. This structure helps substantially in this non-selfadjoint setting, and using also
an additional symmetry of a PT--type, we were able to compute the singular values, in order to get results in the form given in (\ref{mu2}).
Let us also remark that if the supersymmetric method is not available, a natural approach to the tunneling analysis would proceed by
means of exponentially weighted estimates of Agmon type, where the notion of the Agmon distance should be replaced by a degenerate
Finsler distance, see~\cite{LiNi}.

\medskip
\noindent
In fact in \cite{HHS08a}, \cite{HHS08b}, \cite{HHS11}, it was shown, under additional assumptions, that the tunneling of any supersymmetric operator of order two
could be studied using the method given there, and the case of chains of oscillators coupled to two heat baths was given there as an example. This model was
first introduced in \cite{EPR99} and was shown to be supersymmetric in the case of the equal temperatures in \cite{HHS08b}. For general second order
non-negative operators, the question of intrinsic geometric structures and their relations with the spectral gap in asymptotic regimes was also studied in
\cite{FHPS10} in the elliptic case, and this article was one of the motivations for our study of supersymmetric hidden structures here.

\medskip
\noindent
The paper is organized as follows. In Section 2, we first give a necessary and sometimes sufficient condition for a general real partial
differential operator of order two to be supersymmetric, both in the non-semiclassical and the semiclassical cases, and establish Theorem \ref{supsym}.
We then review the three supersymmetric examples and in particular, discuss the equilibrium and the uncoupled cases for the chains (case i) of
Theorem \ref{nosupchains}. Section 3 is devoted to the study of the case of different temperature for the chains and to the proof of part ii) of
Theorem \ref{nosupchains}.

\section{Generalities on supersymmetric structures and some examples}\label{all}
\setcounter{equation}{0}
In the beginning of this section we shall establish Theorem \ref{supsym}. Let us begin with the non-semiclassical case $h=1$.
Let $X$ be $\real^n$ or a smooth compact manifold of dimension $n$, equipped with a smooth strictly positive volume density $\omega (dx)$. Assume that local
coordinates $x_1,..,x_n$ have been chosen so that $\omega (dx)=dx_1\cdot ...\cdot dx_n$.
If
\ekv{all.1}
{v=\sum_{j=1}^n v_j(x)\partial _{x_j}}
is a smooth vector field on $X$, then the divergence of $v$ is well defined by the choice of $\omega$ and in the special coordinates above, we have
\ekv{all.2}{\mathrm{div\,}v=\sum_{j=1}^n \partial _{x_j}(v_j).}

\medskip
\noindent
Let $d$ be the de Rham exterior differentiation and let $\delta$ be its adjoint, defined in the introduction. Then using
the local expression (\ref{ddelta}), we obtain that
$$
\delta v =-\mathrm{div\,}(v).
$$

\medskip
\noindent
Let now $A(x): T_x^*X\to T_xX$ depend smoothly on $x\in X$, and let us consider the bilinear product defined in (\ref{bil}) and the induced action of
$A$ on $k$--forms. We now let $P$ be a second order real differential operator on $X$, which we can write in the special local coordinates above as follows,
\ekv{all.3}
{
P=-\sum_{j,k=1}^n \partial _{x_j}\circ B_{j,k}(x)\circ \partial _{x_k}+\sum_{j=1}^n v_j(x)\circ\partial _{x_j}+v_0.
}
Here $(B_{j,k})$ is symmetric and $B_{j,k}$, $v_j$ and $v_0$ are real-valued and smooth. Viewing $P$ as acting on $0$--forms, we first ask
whether there exists a smooth map $A(x)$ as above, such that
\ekv{all.4}
{
P=d^{A,*}d=\delta A^\mathrm{t}d,
}
either locally or globally on $X$.

\begin{prop}\label{all1}
In order to have {\rm (\ref{all.4})}, it is necessary that \ekv{all.5} {P(1)=0\hbox{ and }P^*(1)=0.  } Here $P^*$ is the adjoint of our
scalar operator $P$ with respect to the $L^2$--scalar product determined by the density of integration $\omega(dx)$.
\end{prop}
\begin{proof}
If (\ref{all.4}) holds, then clearly $P(1)=0$, since $d(1)=0$. Since $P^*=\delta Ad$, we also have $P^*(1)=0$.
\end{proof}

\begin{prop}\label{all2}
The property {\rm (\ref{all.5})} is equivalent to the following property,
\ekv{all.6}
{
v_0=0 \hbox{ and }\mathrm{div\,}v=0.
}
\end{prop}
\begin{proof}
In the special local coordinates, we see that
$$
P^*=-\sum_{j,k=1}^n \partial _{x_j}\circ B_{j,k}(x)\circ \partial _{x_k}-\sum_{j=1}^n v_j\circ \partial _{x_j}-\mathrm{div\,}v+v_0.
$$
Thus the property $P(1)=0$ is equivalent to $v_0=0$ and the property $P^*(1)=0$ is equivalent to $v_0-\mathrm{div\,}(v)=0$.
\end{proof}

\bigskip
\noindent
We now look for a smooth map $A=A(x)$ such that (\ref{all.4}) holds. In the special coordinates above, let $A=(A_{j,k}(x))$ be the
matrix of $A$. Then,
$$
\delta A^{\mathrm{t}}d=-\sum_{j,k=1}^n \partial _{x_j}\circ A_{k,j}(x)\circ \partial _{x_k}.
$$
Write $A=\widetilde{B}+C$, where $\widetilde{B}$ is symmetric, $\widetilde{B}^\mathrm{t}=\widetilde{B}$, and $C$ is antisymmetric, $C^{\mathrm{t}}=-C$. Then
we get
\ekv{all.7} { \delta A^\mathrm{t}d=-\sum_{j,k=1}^n\partial
  _{x_j}\circ \widetilde{B}_{j,k}(x)\circ \partial _{x_k}+\sum_{k=1}^n(\sum_{j=1}^n \partial
  _{x_j}(C_{j,k}))\partial _{x_k}.  }
In order to have (\ref{all.4}), we see that $\widetilde{B}$ has to be equal to $B$, and we assume that from now on, so that
\ekv{all.8}{A=B+C,\quad B^\mathrm{t}=B,\ C^\mathrm{t}=-C,} where $B$ is the matrix appearing in (\ref{all.3})
and the antisymmetric matrix $C$ remains to be determined. In order to have (\ref{all.4}) it is necessary and sufficient that
\ekv{all.9}
{v_k=\sum_{j=1}^n\partial _{x_j}(C_{j,k}).}

\medskip
\noindent
Consider the 2-vector form
$$
\Gamma =\sum_{j,k=1}^nC_{j,k}\partial _{x_j}\wedge\partial _{x_k}.
$$
By a straightforward calculation, see (\ref{ddelta}), we obtain that
\[
\begin{split}
\delta \Gamma &=-\sum_{\nu=1}^n\frac{\partial }{\partial x_\nu }dx_\nu ^\rfloor(\sum_{j,k=1}^nC_{j,k}\partial _{x_j}\wedge\partial _{x_k})\\
&=-2\sum_{j,k=1}^n\frac{\partial }{\partial x_j}(C_{j,k})\partial _{x_k},
\end{split}
\]
so (\ref{all.9}) amounts to
\ekv{all.10}{v=-\frac{1}{2}\delta \Gamma.}
As we have seen, the necessary condition (\ref{all.5}) is equivalent to (\ref{all.6}), which contains the assumption that $\mathrm{div\,}v=0$, i.e.
that $\delta v=0$. Since the $\delta$--complex is locally exact in degree 1, we get the following result.

\begin{prop}\label{all3}
If {\rm (\ref{all.5})} holds, then locally, we can find a smooth matrix $A$ such that {\rm (\ref{all.4})} holds. More precisely, we can find a smooth matrix
$A$ so that {\rm (\ref{all.4})} holds, in any open subset $\widetilde{X}\subset X$, where the $\delta$--complex is exact in degree {\rm 1} for smooth sections.
\end{prop}

\begin{remark}\label{all4}
If $A=B+C$ is a solution to {\rm (\ref{all.4})}, then $A=B+\widetilde{C}$ is another solution if and only if $\widetilde{C}-C$ (identified with a
{\rm 2}-vector form) is $\delta $--closed.
\end{remark}

\bigskip
\noindent
We shall now replace the assumption (\ref{all.5}) by the more general assumption that there exist smooth strictly positive functions
$e^{-\varphi }$ and $e^{-\psi }$ in the kernels of $P$ and $P^*$ respectively,
\ekv{all.11}{P(e^{-\varphi })=0,\ P^*(e^{-\psi })=0.}
Put $\widetilde{P}=e^{-\psi }\circ P\circ e^{-\varphi }$, so that $\widetilde{P}^*=e^{-\varphi }\circ P^*\circ e^{-\psi}$. {\it Then
$\widetilde{P}$ satisfies {\rm (\ref{all.5})}.} Hence if $\widetilde{X}\subset X$ is an open connected subset where the $\delta$--complex
is exact in degree 1, we have a smooth matrix $\widetilde{A}$ on $\widetilde{X}$ such that
\ekv{all.12}
{\widetilde{P}= \delta\widetilde{A}^\mathrm{t}d \hbox{ in }\widetilde{X}.}  Putting
$d_{\varphi} =e^{-\varphi}\circ d \circ e^{\varphi}$, $d_{\psi} =e^{-\psi}\circ  d \circ e^{\psi}$, we get with $A=e^{\varphi +\psi }\widetilde{A}$:
\[
P=e^{\psi} \widetilde{P} e^{\varphi} = e^{\psi} \delta  e^{-\psi }A^\mathrm{t}e^{-\varphi} d e^{\varphi} = d_{\psi} ^*A^\mathrm{t}d_{\varphi} =
(d_{\psi} )^{A,*}d_{\varphi}.
\]

\bigskip
\noindent
Summarizing the discussion so far, we see that we have now established the first part of Theorem \ref{supsym}, addressing the non-semiclassical case. Restoring now
the semiclassical parameter $h\in (0,h_0]$ and recalling the notion of a supersymmetric structure in the semiclassical sense given in Definition 1.1, we see
that the arguments given above can be applied to the conjugated operator $e^{-\psi/h}\circ P(x,hD_x;h) \circ e^{-\varphi/h}$,
where $\varphi=\varphi_0(x) + {\cal O}(h)$, $\psi = \psi_0(x) + {\cal O}(h)$, for each fixed value of $h\in (0,h_0]$. The second statement
of Theorem \ref{supsym} follows and the proof of Theorem \ref{supsym} is therefore complete.

\begin{cor}
Let $P=P(x,hD_x;h)$ be a second order real semiclassical differential operator on $X$, which in local coordinates, can be written as follows,
$$
P = -\sum_{j,k=1}^n h\partial_{x_j}\circ B_{j,k}(x)\circ h\partial_{x_k} + \sum_{j=1}^n v_j(x)\circ h\partial_{x_j} + v_0(x) + hw_0(x).
$$
Here $(B_{j,k})$ is symmetric, and $(B_{j,k})$, $v_j$, $v_0$, $w_0$ are real-valued, smooth, and independent of $h$. Then a necessary condition for the operator
$P$ to be supersymmetric in the semiclassical sense is that there should exist smooth $h$--independent functions $\varphi_0$ and $\psi_0$ satisfying the following
eikonal equations on $X$,
$$
\sum_{j,k=1}^n B_{j,k}(x) \partial_{x_j}\varphi_0 \partial_{x_k}\varphi_0 +\sum_{j=1}^n v_j(x) \partial_{x_j} \varphi_0(x) - v_0(x) = 0,
$$
and
$$
-\sum_{j,k=1}^n B_{j,k}(x) \partial_{x_j}\psi_0 \partial_{x_k}\psi_0 +\sum_{j=1}^n v_j(x) \partial_{x_j} \psi_0(x) + v_0(x) = 0.
$$
\end{cor}

\medskip
\noindent
{\it Remark}. Let us notice that if an operator $P$ possesses a supersymmetric structure in the semiclassical sense, then so does
$e^{f/h}\circ P \circ e^{-f/h}$, for a smooth $h$--independent function $f$. Indeed, if $P = d_{\psi,h}^{A,*} d_{\varphi,h}$, then
$$
e^{f/h}\circ P \circ e^{-f/h} = d_{\psi+f,h}^{A,*} d_{\varphi-f,h}.
$$

\bigskip
\noindent
We shall end this section by coming back to the three examples mentioned in the introduction, namely, the Witten Laplacian, the Kramers-Fokker-Planck operator,
and the model of chains of oscillators coupled to heat baths. All three of them come from stochastic differential equations, and we refer to Section 5 in
\cite{HHS08b} for a complete discussion concerning their derivation. Here we shall merely quote some results given there. The associated time-dependent equation,
\begeq
\label{eqf}
\left(h\partial_t + {P}\right)f(t,x)=0,\quad t\geq 0,
\endeq
describes the evolution of the particle density, and due to the stochastic origins of this model, the equation $P^* (1) = 0$ is always satisfied, so that for
$\psi=0$, the equation $P^*(e^{-\psi/h}) = 0$ is given for free. According to Theorem \ref{supsym}, the only remaining thing to check in order
to obtain the existence of a semiclassical supersymmetric structure for $P$ is then the existence of a smooth function $\varphi(x;h)=\varphi_0(x) + {\cal O}(h)$,
such that $P(e^{-2\varphi/h}) = 0$, for all $h>0$ small enough.

\paragraph{The Witten Laplacian.}
We begin with the semiclassical Witten case. It corresponds to an evolution equation with a gradient field $ -\gamma \nabla V(x)$ and a diffusion
force coming from a heat bath at a temperature $T=h/2$. The stochastic differential equation corresponding to this model is of the form,
$$
dx = - \gamma \D_x V dt + \sqrt{\gamma h} d w.
$$
Here $x\in \real^n$ is the spatial variable, the parameter $\gamma>0$ is a friction coefficient, and $w$ is an
$n$--dimensional Wiener process of mean $0$ and variance $1$.   Equation (\ref{eqf}) for the particle density in this case is then
\begin{equation}
\label{W1}
h\D_t f -\frac{\gamma}{2} h\D_x\cdot (h\D_x + 2\D_x V) f = 0.
\end{equation}
We see that $P^*(1) = 0$, and also that $P(e^{-2\varphi/h}) = 0$, for all $h>0$. Here we have written $\varphi=V$.
An application of Theorem \ref{supsym} gives that $P$ has a supersymmetric structure on $\real^n$ in the semiclassical sense.
The function $\mmm(x) = e^{-2\varphi(x)/h}$ is sometimes called the (non-normalized) Maxwellian of the system and corresponds to a stationary density, provided
that it is integrable on $\real^n$.

\medskip
\noindent
Writing $f = \mmm^{1/2} u$, we obtain from (\ref{W1}) that
\begin{equation} 
 h\D_t u + \frac{\gamma}{2} ( - h\D_x + \D_x V)\cdot (h\D_x + \D_x V) u = 0,
 \end{equation}
and  here we recognize the Witten operator $W = ( - h\D_x + \D_x V)\cdot (h\D_x + \D_x V)$ of the introduction. The supersymmetric structure of $W$
is given by
$$
A= \frac{\gamma}{2} {\rm I}_n, \ \ \ \varphi(x) = \psi(x) = V(x).
$$

\medskip
\paragraph{The Kramers-Fokker-Planck operator.}
The stochastic diffe\-ren\-tial equa\-ti\-on associated to this model is
\begin{equation}
\left\{ \begin{array}{l} dx = y dt \\ dy =  - \gamma y dt- \D_x V(x) dt
+ \sqrt{\gamma h} d w. \end{array} \right.
\end{equation}
The parameter $\gamma>0$ is a friction coefficient, and the particle of
position $x \in \real^n$ and velocity $y\in \real^n$ is submitted to an external force
field derived from a potential $V$, with $w$ being an $n$--dimensional
Wiener process of mean $0$ and variance $1$.
The corresponding equation for the particle density (\ref{eqf}) is then
\begin{equation}
\label{KFPn}
 h\D_t f -\frac{\gamma}{2} h\D_y\cdot ( h \D_y + 2y) f +
 y \cdot h\D_x f -  \D_x V \cdot h\D_y  f  = 0.
 \end{equation}
We have $P^*(1) = 0$, and posing $\varphi(x,y) = {y^2}/2 + V(x)$ we also see that $P(e^{-2\varphi/h}) = 0$. An application of Theorem \ref{supsym} gives
that $P$ is semiclassically supersymmetric on $\real^{2n}$. The (non-normalized) Maxwellian is given by $\mmm(x,y) =  e^{-2\varphi(x,y)/h}$.

\medskip
\noindent
If we write $f = \mmm^{1/2} u$, then (\ref{KFPn}) becomes,
\begin{equation} 
h\D_t u +\frac{\gamma}{2}(-h\D_y + y)\cdot ( h \D_y + y) u +
y \cdot h\D_x u -  \D_x V \cdot h\D_y  u  = 0
\end{equation}
Taking $\gamma=2$, we have thus arrived at the semiclassical Kramers-Fokker-Planck operator $K$, mentioned in the introduction,
$$
K = y\cdot h\D_x - \D_x V(x) \cdot h \D_y - h^2 \Delta_y + y^2 - hn.
$$
The supersymmetric structure of $K$ is given by
$$
 A =  \frac{1}{2} \left( \begin{array}{cc} 0 & {\rm I}_n \\ - {\rm I}_n & \gamma  \end{array}\right)
\ \ \ \ \textrm{ and } \ \ \ \  \varphi(x,y) = y^2/2 + V(x), \ \ \ \ \psi(x,y) = y^2/2 + V(x).
$$

\medskip
\paragraph{Chains of oscillators.}
This is a model for a system of particles  described by their respective position and velocity
$(x_j, y_j) \in \real^{2n}$ corresponding to two oscillators.
We suppose that for each oscillator $j \in \set{ 1,2}$, the particles
are  submitted to an external force derived from an effective potential
$W_j(x_j)$, and that there is a coupling between the two oscillators
derived from an effective potential $\delta W (x_1, x_2)$.  We denote by $W$ the sum
$$
W(x) = W_1(x_1)+ W_2(x_2) + \delta W (x_1, x_2),
$$
where $x = (x_1, x_2)$, and we also write $y = (y_1,y_2)$.
By $z_j$, $j \in \set{ 1,2}$ we shall denote the variables describing
the state of the particles in each of the heat baths, and set $z = (z_1,z_2)$. We suppose that the
particles in each bath are submitted to a coupling with the nearest
oscillator, a friction force  and a thermal diffusion at a temperature
$T_j = \alpha_j h /2 $, $j=1,2$. We denote by $w_j$, $j \in \set{ 1,2}$, two
$n$-dimensional Brownian motions of mean 0 and variance $1$, and set $w=(w_1, w_2)$.
The stochastic differential equation for this model is the following,
\begin{equation}
\left\{
\begin{array}{l}
 dx_1 = y_1 dt \\
  dy_1 =  - \D_{x_1} W(x) dt + (z_1-x_1) dt \\
  d z_1 = - \gamma z_1 dt + \gamma x_1 dt - \sqrt{\gamma \alpha_1 h } d
  w_1 \\
  d z_2 = - \gamma z_1 dt + \gamma x_2 dt - \sqrt{\gamma \alpha_2 h} d
  w_2 \\
dy_2 =  - \D_{x_2} W(x) dt + (z_2-x_2) dt \\
d x_2 =y_2 dt.
\end{array} \right.
\end{equation}
The parameter $\gamma>0$  is  the  friction coefficient in the baths.
The corresponding semiclassical equation (\ref{eqf}) for the particle density is then
\begin{equation} \label{firstW6} \begin{split}
 h\D_t f + \widetilde{P}_Wf := h\D_t f & +  \frac{\gamma}{2} \alpha_1 (-h\D_{z_1})( h \D_{z_1} + 2(z_1-x_1)/\alpha_1)f \\
 & + \frac{\gamma}{2} \alpha_2 (-h\D_{z_2})( h \D_{z_2} + 2(z_2-x_2)/\alpha_2)f \\
  &  + (y \cdot h\D_x f -  (\D_x W(x) + x-z) \cdot h\D_y) f  = 0.
   \end{split}
 \end{equation}
We have $\widetilde{P}_W^*(1) = 0$. In order to exhibit a semiclassical supersymmetric structure for $\widetilde{P}_W$, we need to find a smooth function
$\varphi(x;h) = \varphi_0(x)+{\cal O}(h)$, such that $\widetilde{P}_W(e^{-2\varphi/h}) = 0$, for all $h>0$ small enough. With a general configuration,
corresponding to different temperatures and non-trivial coupling $\delta W \neq 0$, the existence of a Maxwellian of the form $e^{-2\varphi/h}$, with
$\varphi$ as above is not clear, and in fact we shall show in the next section that in some cases, there is no such a Maxwellian, thus establishing part ii) of
Theorem \ref{nosupchains}. In the remainder of this section, we shall be concerned with two cases for the model (\ref{firstW6}), where a
supersymmetric structure can be found.

\medskip
\noindent
Let us set $W_0(x) = W_1(x_1) + W_2(x_2)$. The smooth function
\begeq
\label{eq2}
\varphi_0(x,y,z) = \sum_{j=1}^2 \frac{1}{\alpha_j} \left(\frac{y_j^2}{2} + W_j(x_j) + \frac{(x_j-z_j)^2}{2}\right).
\endeq
will be of central importance in the following discussion.

\bigskip
\noindent
\preuve[ of i) Theorem {\rm \ref{nosupchains}} in the case of equal temperatures]
Let us discuss first the case when $\alpha_1=\alpha_2=:\alpha$. In this case it is immediate to check that if we define
$$
\varphi = \varphi_0 + \frac{1}{\alpha}\delta W =
\sum_{j=1}^2 \frac{1}{\alpha_j} \left(\frac{y_j^2}{2} + W_j(x_j) + \frac{(x_j-z_j)^2}{2}\right) + \frac{1}{\alpha} \delta W(x),
$$
then we have  $\widetilde{P}_W(e^{-2\varphi/h}) = 0$. An application of Theorem \ref{supsym} then gives that the operator $\widetilde{P}_W$ is
supersymmetric on $\real^{6n}$, in the semiclassical sense. The associated Maxwellian is then defined up to a constant by
$$
\mmm(x,y,z) =  e^{-2\varphi/h}.
$$
Let us also give the conjugated version of the time-dependent problem. If we write
$$
f = \mmm^{1/2} u,
$$
the equation (\ref{firstW6}) becomes
\begin{equation} 
\begin{split}
 h\D_t u & +  \frac{\gamma}{2} \alpha \sep{-h\D_{z_1} +
 \frac{1}{\alpha}(z_1-x_1)}
             \sep{ h \D_{z_1} + \frac{1}{\alpha}(z_1-x_1)}u \\
 & +  \frac{\gamma}{2} \alpha \sep{-h\D_{z_2} + \frac{1}{\alpha}(z_2-x_2)}
             \sep{ h \D_{z_2} + \frac{1}{\alpha}(z_2-x_2)}u \\
  &  + \sep{y\cdot h \D_x  -  (\D_x W(x) +x-z)\cdot h \D_y} u  = 0.
   \end{split}
 \end{equation}
The conjugated operator,
\begin{multline}
\label{Hp}
{P}_W = \frac{\gamma}{2}\sum_{j=1}^2 \alpha\left(-h\partial _{z_j}+\frac{1}{\alpha}(z_j-x_j)\right)\left(h\partial
_{z_j}+\frac{1}{\alpha}(z_j-x_j)\right) \\
 + y\cdot h\partial_x -(\partial_x W(x)+x-z)\cdot h\partial _y
\end{multline}
then enjoys a supersymmetric structure of the form ${P}_W = d_{\varphi,h}^{A,*} d_{\varphi,h}$, with
\begin{equation}
\label{phibb}
\varphi =   \varphi_0 + \frac{1}{\alpha} \delta W = \frac{1}{\alpha} \sep{W(x) + y^2/2 + (z-x)^2/2 },
\end{equation}
and a non-degenerate (constant) matrix $A$   given by
\begin{equation} \label{same}
A = \frac{\alpha}{2} \left( \begin{array}{ccc}
0 & {\mathrm I}_n & 0 \\
-{\mathrm I}_n & 0 & 0 \\
0 & 0 & \gamma {\mathrm I}_n
\end{array}\right) .
\end{equation}

\bigskip
\preuve[ of i) Theorem {\rm \ref{nosupchains}} in the case when $\delta W \equiv 0$]
We consider now the case when $\delta W\equiv 0$, so that $W=W_0$, while $\alpha_1 $ and $\alpha_2$ may be unequal. In this case it is immediate to check that with
$\varphi_0$ defined in (\ref{eq2}), we have  $\widetilde{P}_{W_0}(e^{-2\varphi_0/h}) = 0$. Following Theorem \ref{supsym}, we conclude therefore that
the operator $\widetilde{P}_{W_0}$ is supersymmetric, in the semiclassical sense. The Maxwellian associated to the problem is then defined up to a constant by
$$
\mmm(x,y,z) =  e^{-2\varphi_0/h}.
$$
If we write as before,
$$
f = \mmm^{1/2} u,
$$
then the equation (\ref{firstW6}) becomes
\begin{equation} \label{firstW7} \begin{split}
 h\D_t u & +  \frac{\gamma}{2} \alpha_1 \sep{-h\D_{z_1} +
 \frac{1}{\alpha_1}(z_1-x_1)}
             \sep{ h \D_{z_1} + \frac{1}{\alpha_1}(z_1-x_1)}u \\
 & +  \frac{\gamma}{2} \alpha_2 \sep{-h\D_{z_2} + \frac{1}{\alpha_2}(z_2-x_2)}
             \sep{ h \D_{z_2} + \frac{1}{\alpha_2}(z_2-x_2)}u \\
  &  + \sep{y\cdot h \D_x  -  (\D_x W_0+ x  -z) \cdot h \D_y} u  = 0.
   \end{split}
 \end{equation}
We then have a supersymmetric structure for the conjugated operator, occurring in (\ref{firstW7}), with
\begin{equation} \label{phibbb}
\varphi = \psi =  \varphi_0  = \sum_{j=1}^2 \frac{1}{\alpha_j} \sep{W_j(x) + y_j^2/2 + (z_j-x_j)^2/2 },
\end{equation}
and a non-degenerate (constant) matrix $A$   given by
$$
A = \frac{1}{2} \left( \begin{array}{cccccc}
0&0  & \alpha_1&0 & 0&0 \\
0&0  & 0&\alpha_2  & 0&0 \\
-\alpha_1&0 & 0&0 & 0&0 \\
0&-\alpha_2 & 0&0 & 0&0 \\
0&0 & 0&0 & \gamma \alpha_1&0 \\
0&0 & 0&0 & 0&\gamma \alpha_2 \\
\end{array}\right) .
$$

\section{Non-existence of super\-symmetric struc\-tures for chains of os\-cil\-lators}
\setcounter{equation}{0}
The purpose of this section is to discuss the question of non-existence of smooth supersymmetric structures for the model of a chain of
anharmonic oscillators, in the case when the temperatures of the baths are unequal and the effective interaction potential $\delta W$ is non-vanishing.
This will establish part ii) of Theorem \ref{nosupchains}.

\medskip
\noindent
Following (\ref{firstW6}), we are interested in the equation,
\begeq
\label{eq0}
\left(h\partial_t + \widetilde{P}_W\right)f(t,x,y,z)=0,\quad t\geq 0,
\endeq
where
\begin{equation}
\label{firstW6bis}
\widetilde{P}_W = \frac{\gamma}{2}\sum_{j=1}^2 \alpha_j(-h\partial _{z_j})\left(h\partial_{z_j}+\frac{2}{\alpha _j}(z_j-x_j)\right) + y\cdot h\partial_x -
(\partial_x W(x)+x-z)\cdot h\partial _y.
\end{equation}
Here the equation $\widetilde{P}_W^*(1) = 0$ is satisfied automatically. In order to establish the breakdown of the semiclassical supersymmetry for
$\widetilde{P}_W$, we only have to show that there does not exist a function $\varphi\in C^{\infty}(\real^{6n};\real)$ with
$\varphi(x;h) = \varphi_0(x) + {\cal O}(h)$ in the $C^{\infty}$--sense, such that $\widetilde{P}_W(e^{-2\varphi/h}) = 0$ for all $h>0$ small enough,
for a suitable choice of the effective potentials $W_1$, $W_2$, and $\delta W$.

\bigskip
\noindent
In what follows, when $W=W_0$, we shall find it convenient to work with the conjugated operator
\begeq
\label{eq3}
P_0=P_{W_0} := e^{\varphi_0/h} \circ \widetilde{P}_{W_0} \circ e^{-\varphi_0/h},
\endeq
where we recall that the function $\varphi_0$ is given by
$$
\varphi_0  = \sum_{j=1}^2 \frac{1}{\alpha_j} \left(\frac{y_j^2}{2} + W_j(x_j) + \frac{(x_j-z_j)^2}{2}\right)
$$
Following (\ref{firstW7}), we obtain that
\begin{multline}
\label{eq4}
P_{0} = \frac{\gamma}{2}\sum_{j=1}^2 \alpha_j\left(-h\partial _{z_j}+ \frac{1}{\alpha_j}(z_j-x_j)\right)
\left(h\partial_{z_j}+\frac{1}{\alpha _j}(z_j-x_j)\right) \\
+ y\cdot h\partial_x - (\partial_x W_0(x)+x-z)\cdot h\partial _y,
\end{multline}
with the function $e^{-\varphi_0/h}$ being in the kernel of $P_{0}$ and its adjoint.

\bigskip
\noindent
The leading part $p_{0}=p_{0}(x,y,z,\xi,\eta, \zeta)$ of the semiclassical symbol of $P_{0}=P_{0}(x,y,z,hD_x,hD_y, hD_z;h)$ is given by
\begeq
\label{eq5}
p_{0} = \frac{\gamma}{2} \sum_{j=1}^2 \alpha_j \left(\zeta_j^2 + \frac{1}{\alpha_j^2} (x_j-z_j)^2\right) + iy\cdot \xi - i(\partial_x W_0(x) + x -z)\cdot \eta.
\endeq
We may notice here that $\mathrm{Re}\, p_{0} \geq 0$. Associated to the operator $P_{0}$ in (\ref{eq4}) is the real symbol
$$
q_{0}(x,y,z,\xi,\eta,\zeta) = -p_0(x,y,z,i\xi,i\eta,i\zeta),
$$
with
$$
P_0(x,y,z, hD_x,hD_y,hD_z) = -q_0(x,y,z,-h\partial_x, -h\partial_y, -h\partial_z)
$$
to leading order. We have
\begeq
\label{eq6}
q_0 = \frac{\gamma}{2} \sum_{j=1}^2 \alpha_j \left(\zeta_j^2 - \frac{1}{\alpha_j^2} (x_j-z_j)^2\right) + y\cdot \xi - (\partial_x W_0(x) + x-z)\cdot \eta.
\endeq
The phase function $\varphi_0\in C^{\infty}(\real^{6n})$ satisfies the eikonal equation
\begeq
\label{eq7}
q_0(x,y,z,\partial_x \varphi_0, \partial_y \varphi_0, \partial_z \varphi_0) = 0,
\endeq
reflecting the fact that $P_0(x,y,z,-ih\partial_x,-ih\partial_y,-ih\partial_z;h) \left(e^{-\varphi_0/h}\right)=0$.

\bigskip
\noindent
Let $x_0\in \real^{2n}$ be a non-degenerate critical point of $W_0(x)$, and let $\rho \in T^*\real^{6n}$ be the corresponding point in the phase space, given by
$x=x_0$, $y=0$, $z=x_0$, $\xi=\eta = \zeta =0$. Let furthermore $F_{p_0} = F_{p_0,\rho}$ be the fundamental matrix of the quadratic approximation of $p_0$
at the doubly characteristic point $\rho$. Then, as explained in~\cite{HHS08a},~\cite{HHS08b}, the spectrum of $F_{p_0}$ avoids the real axis and is of the form
$\pm \lambda_k$, $1\leq k \leq 6n$, $\textrm{Im}\, \lambda_k > 0$.

\medskip
\noindent
{\it Remark.} The preceding result comes from the fact that near $\rho$, the average of $\textrm{Re}\, p_0$ along the Hamilton flow of
$\textrm{Im}\, p_0$ has a positive definite quadratic part in its Taylor expansion at $\rho$, see \cite{HHS08a}, \cite{HP09}.
Following \cite{HP09} and \cite{OtPaPr}, this can also be interpreted by saying that the so-called singular space $S$, defined by
$$
S =\Big(\bigcap_{j=0}^{\infty}\textrm{Ker}\big[\textrm{Re }F_{p_0}(\textrm{Im }F_{p_0})^j \big]\Big) \bigcap T^*\real^{6n},
$$
is trivial in this case.

\bigskip
\noindent
Let next $\rho \in T^*\real^{6n}$ and $F_{p_0}$ be as above, and let $F_{q_0}$ be the fundamental matrix of the quadratic approximation of $q_0$ at the point
$\rho$. Following e.g. \cite{HHS08a}, we see that the eigenvalues of $F_{q_0}$ and $i^{-1} F_{p_0}$ are the same and are of the form $\pm \frac{1}{i} \lambda_k$, $1\leq k \leq 6n$, where
$$
\textrm{Re}\, \left(\frac{1}{i} \lambda_k\right) > 0, \quad 1\leq k \leq 6n.
$$
As explained in~\cite{HHS08a}, an application of the stable manifold theorem allows us to conclude that in a suitable neighborhood of the point
$(x_0,0,x_0) = \pi(\rho)$, the eikonal equation
\begeq
\label{eq8}
q_0(x,y,z,\partial_x \varphi, \partial_y \varphi, \partial_z \varphi)=0,
\endeq
has a unique smooth solution $\varphi(x,y,z)$ such that $\varphi(\pi(\rho)) = \varphi'(\pi(\rho))=0$ and $\varphi''(\pi(\rho))$ is positive definite. Here
$\pi((x,y,z,\xi,\eta,\zeta)) = (x,y,z)$. Indeed, the function $\varphi(x,y,z)$ is obtained as the generating function for the unstable manifold through
$\rho$ for the $H_{q_0}$--flow. Applying this discussion to the case when $x_0$ is a non-degenerate local minimum of $W_0(x)$, we see that necessarily, we have
the equality $\varphi = \varphi_0$ in a neighborhood of $\pi(\rho)$.

\bigskip
\noindent
We shall assume throughout this section that $W_0(x) = W_1(x_1) + W_2(x_2)$, where $W_2$ is a positive definite quadratic form on
$\real^n$. As for $W_1$, we assume that $W_1$ is a Morse function on $\real^n$ such that $\partial^{\alpha}_{x_1} W_1 = {\cal O}(1)$ for all
$\alpha\in \nat^n$ with $\abs{\alpha}\geq 2$, and such that $\abs{\nabla W_1(x_1)} \geq 1/C$ for $\abs{x_1} \geq C>0$. Assume furthermore that $W_1$ is a double
well potential, so that it has precisely three critical points: two local minima $m_{\pm}$ and a saddle point $s_0$ of signature $(n-1,1)$. The critical points of
$W_0$ are then the local minima $M_{\pm} = (m_{\pm},0)$ and the saddle point $S_0 = (s_0,0)$, of signature $(2n-1,1)$. Furthermore, $W_0(x)\rightarrow +\infty$ as
$\abs{x} \rightarrow \infty$. From (\ref{eq2}) we see that the restriction of $\varphi_0$ to the subspace
$$
L=\{ (x,y,z)\in {\bf R}^{6n};\ z=x,\ y=0\}
$$
can be identified with $W_0$ and that the Hessian of $\varphi_0$ in the directions orthogonal to $L$ is positive definite. Consequently,
$\varphi_0$ is also a Morse function on $\real^{6n}$, tending to $+\infty $ when $(x,y,z)\to \infty $, with precisely three critical points, given by
the local minima $\widetilde{M}_{\pm} = (M_{\pm},0,M_{\pm})$ and the saddle point $\widetilde{S}_0 = (S_0,0,S_0)$, of signature $(6n-1,1)$.

\bigskip
\noindent
Let $\delta W(x)\in C^{\infty}(\real^{2n})$, and thinking about $\delta W(x)$ as a perturbation, let us introduce the perturbed effective potential
$W = W_0 + \delta W$. The corresponding operator $\widetilde{P}_W$ in (\ref{firstW6bis}) is then of the form,
\begeq
\label{eq8.0}
\widetilde{P}_W = \widetilde{P}_{W_0} - \partial_x\left(\delta W(x)\right)\cdot h\partial_y,
\endeq
and after a conjugation, we obtain the operator
\begin{multline}
\label{conj}
P_W := e^{\varphi_0/h} \circ \widetilde{P}_{W} \circ e^{-\varphi_0/h} = P_0 -
e^{\varphi_0/h} \circ \partial_x\left(\delta W(x)\right)\cdot h\partial_y \circ e^{-\varphi_0/h} \\
= P_0 - \partial_x\left(\delta W(x)\right)\cdot (h\partial_y - \partial_y \varphi_0).
\end{multline}
Associated to the operator (\ref{conj}), we have the perturbed real Hamiltonian,
\begeq
\label{eq8.1}
q(x,y,z,\xi,\eta,\zeta) = q_0(x,y,z,\xi,\eta,\zeta) - \partial_x \left(\delta W(x)\right)\cdot \left(\eta+\partial_y \varphi_0\right).
\endeq

\bigskip
\noindent
We are interested in the question whether the perturbed conjugated operator $P_W$ still possesses a smooth supersymmetric structure on $\real^{6n}$,
in the semiclassical sense. According to Corollary 2.5, a necessary condition for that is the existence of a smooth solution $\varphi$ of the eikonal equation,
\begeq
\label{eq8.2}
q(x,y,z,\partial_x\varphi,\partial_y\varphi, \partial_z\varphi) = 0.
\endeq

\medskip
\noindent
The following is the main result of this section.

\begin{theo}
\label{thm_main}
With a suitable choice of $W_1$ and $W_2$ as above, there exists $\delta W\in C^{\infty}(\real^{2n})$ with $M_+$ and
$S_0 \not\in \textrm{supp} (\delta W)$ such that the eikonal equation {\rm (\ref{eq8.2})}
does not have any solution $\varphi\in C^{\infty}(\Omega)$, with $\varphi(\widetilde{M}_+)=\varphi'(\widetilde{M}_+)=0$, $\varphi''(\widetilde{M}_+)>0$,
for any open set $\Omega\subset \real^{6n}$ such that $\Omega_+ \Subset \Omega$. Here $\Omega_+$ is the connected component of the set
$\varphi_0^{-1} ((-\infty, \varphi_0(\widetilde{S}_0)))$, which contains $\widetilde{M}_+$.
\end{theo}

\begin{cor} The perturbed operator $\widetilde{P}_W$ in {\rm (\ref{eq8.0})} does not possess any smooth supersymmetric structure
in the open set $\Omega$, in the semiclassical sense.
\end{cor}

\medskip
\noindent
{\it Remark}. The idea behind the result of Theorem \ref{thm_main} is that the solvability of the eikonal equation (\ref{eq8.2}) does not seem to be an issue
near the local minimum $\widetilde{M}_+$, but if we start by solving the problem near this point and try to extend the solution, we may run into some
trouble when approaching the saddle point $\widetilde{S}_0$.

\bigskip
\noindent
When establishing Theorem \ref{thm_main}, it will be convenient to write $\varphi = \varphi_0 + \psi$, and to consider the corresponding eikonal
equation for $\psi$. A direct computation using (\ref{eq2}), (\ref{eq6}), and (\ref{eq8.1}) shows that $\psi\in C^{\infty}(\Omega)$ should satisfy the
following equation,
\begeq
\label{eq9}
\nu \psi + \frac{\gamma}{2} \sum_{j=1}^2 \alpha_j (\partial_{z_j} \psi)^2 - \partial_x \delta W \cdot \partial_y \psi =
2 \partial_x \delta W \cdot \partial_y \varphi_0.
\endeq
Here the real vector field $\nu$ is given by
\begeq
\label{eq9.5}
\nu = \gamma(z-x)\cdot \partial_z + y\cdot\partial_x - (\partial_x W_0(x) + x-z)\cdot \partial_y,
\endeq
and we notice that $\nu$ can be identified with the Hamilton vector field $H_{q_0}$ along $\Lambda_{\varphi_0}$,
via the projection $\pi|_{\Lambda_{\varphi_0}}$. Here the Lagrangian manifold $\Lambda_{\varphi_0}\subset T^*\real^{6n}$ is given by
$(\xi,\eta,\zeta) = (\partial_x \varphi_0, \partial_y \varphi_0, \partial_z \varphi_0)$.

\bigskip
\noindent
We shall prepare for the analysis of the eikonal equation (\ref{eq9}) by studying the behavior of the integral curves of $\nu$, in particular close
to the stationary points of $\nu$. Here we may also notice that the flow of $\nu$ is complete.

\medskip
\noindent
Without using the assumption that $W_0$ is a direct sum, we see that the vector field $\nu$ vanishes at a point $(x,y,z)$ precisely when
\begeq
\label{eq9.6}
y=0,\ z=x,\ \partial _x W_0(x)=0.
\endeq
Recall that for simplicity we assumed  $\gamma =1$, and let $x_0$ be a critical point of $W_0$. Then at the corresponding point
$(x_0,y_0,z_0)=(x_0,0,x_0)$, the linearization of $\nu$ has the block matrix
\begeq
\label{eq9.7}
N=\begin{pmatrix}0 &1 &0\\ -W''(x_0)-1 &0 &1\\ -1 &0 &1\end{pmatrix}.
\endeq
Let $\lambda$ be an eigenvalue of $N$ and let $(x\ y\ z)^t$ be a corresponding eigenvector. Then we get the system
$$
\begin{cases}-\lambda x+y &=0,\\
(W''(x_0)+1)x+\lambda y-z&=0,\\
-x+(1-\lambda )z&=0.
\end{cases}
$$
Noticing that $\lambda =1$ is not an eigenvalue, we obtain that
$$
y=\lambda x,\ z=(1-\lambda )^{-1}x,
$$
where $x$ satisfies
$$
W''(x_0) x = \left(\frac{\lambda }{1-\lambda }-\lambda ^2 \right)x.
$$
Thus, each eigenvalue $w$ of $W''(x_0)$ gives rise to three eigenvalues $\lambda $ of $N$, given by the equation
\begeq
\label{eq9.8}
F(\lambda):=\frac{\lambda}{1-\lambda}-\lambda ^2=w.
\endeq
We have
\begeq
\label{eq9.9}
F'(\lambda) = \frac{G(\lambda)}{(1-\lambda)^2}, \quad \lambda \neq 1,
\endeq
where $G(\lambda) = 1 - 2\lambda(1-\lambda)^2$. It is easy to see that $G(\lambda)>0$ for $\lambda < 1$, and therefore,
$F(\lambda)$ is strictly increasing from $-\infty$ to $+\infty$ on the interval $(-\infty,1)$, with $F(0)=0$. Next, $G(\lambda)$ is strictly decreasing from $1$ to
$-\infty$ on the interval $(1,\infty)$, and therefore, on this interval, the function $F$ has a unique critical point $m>1$, with $F(m)<0$.
(Numerical computations show that $m=1,5652$.)

\medskip
\noindent
The equation (\ref{eq9.8}) can be written as a polynomial equation $\lambda ^3-\lambda^2+(1+w)\lambda -w=0$, so that the three roots $\lambda_j$ satisfy
\begeq
\label{eq9.91}
\sum_{j=1}^3 \lambda_j=1.
\endeq
We get therefore the following information about the solutions of (\ref{eq9.8}).
\begin{itemize}
\item If $w\le F(m)$, then the three roots are real, and two are $>1$. The remaining root is negative.
\item If $w>F(m)$, then we have precisely one real root $\lambda_1$ and it belongs to the interval $(-\infty ,1)$. The other two roots are of the form
$a\pm ib$, with $a$, $b\in \real$, and from the relation (\ref{eq9.91}) we see that $2a+\lambda _1=1$, implying that $a>0$.
\end{itemize}
The main conclusion concerning the three roots $\lambda _j$ is as follows:
\begin{itemize}
\item If $w>0$ then all three roots have  positive real parts.
\item If $w=0$ then one root vanishes and the other two have positive
  real parts.
\item If $w<0$ then one root is negative and the other two have  positive
  real parts.
\item For each eigenvalue $w$, the $x$-component of the eigenvector $(x,y,z)^t$ of $N$ is the corresponding eigenvector of $W''(x_0)$.
\end{itemize}

\bigskip
\noindent
As the next step in the analysis of the $\nu$-flow, we shall show that $\varphi_0$ is strictly increasing along the integral curves of $\nu$ in the
region where $W_0'(x)\ne 0$. When doing so, it will be convenient to write
$$
\nu = \nu_1(w_1,\partial_{w_1}) + \nu_2(w_2,\partial_{w_2}),\quad w_j = (x_j,y_j,z_j).
$$
Using (\ref{eq2}) and (\ref{eq9.5}), we get, still assuming that $\gamma =1$,
\begeq
\label{eq9.92}
\nu (\varphi_0)=\sum_{j=1}^2\frac{1}{\alpha _j}(z_j-x_j)^2,
\endeq
\begeq
\label{eq9.93}
\nu ^2(\varphi_0 )=\sum_{j=1}^2 \frac{2}{\alpha _j}(z_j-x_j)\nu _j(z_j-x_j),
\endeq
\begeq
\label{eq9.94}
\nu^3(\varphi_0)=\sum_{j=1}^2\frac{2}{\alpha _j}\left((\nu_j(z_j-x_j))^2 +(z_j-x_j)\cdot \nu _j^2(z_j-x_j)\right),
\endeq
\begeq
\label{eq9.95}
\nu ^4(\varphi_0)=\sum_{j=1}^2\left(\frac{6}{\alpha _j}(\nu _j^2(z_j-x_j))\nu_j(z_j-x_j)+\frac{2}{\alpha _j}(z_j-x_j)\nu _j^3(z_j-x_j) \right),
\endeq
\begin{multline}
\label{eq9.96}
\nu ^5(\varphi_0 )=\sum_{j=1}^2 \left(\frac{6}{\alpha _j}(\nu
  _j^2(z_j-x_j))^2+\right.\\
\left.+ \frac{8}{\alpha _j}(\nu _j^3(z_j-x_j))\nu _j(z_j-x_j)+\frac{2}{\alpha _j}(z_j-x_j)\nu _j^4(z_j-x_j) \right).
\end{multline}

\medskip
\noindent
Here
$$
\nu (z-x)=\gamma (z-x)-y,
$$
$$
\nu ^2(z-x)=(\gamma ^2-1)(z-x)-\gamma y+W_0'(x).
$$
Still working in the region where $W'_0(x)\neq 0$, we see that
\begin{itemize}
\item[1)] $\nu (\varphi_0) \ge 0$ with the equality precisely when $z=x$.
\item[2)] If $\nu (\varphi_0)=0$ and $y\ne 0$, then $\nu^3(\varphi_0)>0$.
\item[3)] If $\nu (\varphi_0)=0$ and $y=0$, then $\nu^2(z-x)\ne 0$ and
  $\nu ^5(\varphi_0)>0$.
\end{itemize}
This leads to the statement that for every compact set where $W_0'(x)\ne 0$ there exists a constant $C>0$ such that,
\begeq
\label{eq9.97}
\nu (\varphi_0)\ge \frac{1}{C}\hbox{ or }\nu^3 (\varphi_0)\ge \frac{1}{C}\hbox{ or }\nu^5 (\varphi_0)\ge \frac{1}{C}.
\endeq

\begin{prop}
\label{prop1}
For every compact set $K\subset {\bf R}^{6n}_{x,y,z}$ where $W_0'(x)\ne 0$ there exists a constant $C>0$ such that
\begeq
\label{eq9.98}
\varphi_0 \circ \exp t\nu(x) -\varphi_0 (x)\ge t^5/C,\quad x\in K,\quad 0\le t\le 1/C.
\endeq
\end{prop}
\begin{proof}
Consider the function $f(s)=\nu (\varphi_0)\circ \exp s\nu (x)$ along an integral curve $(-\frac{1}{C},\frac{1}{C})\ni t\mapsto \exp t\nu (x)$
of $\nu $. If $f(0)\ge \mathrm{Const.\,}>0$, we get $\int_0^t f(s)ds\ge t/C$ for $0\le t\ll 1$ and (\ref{eq9.98}) follows (for this
value of $x$).

\medskip
\noindent
If $0\le f(0)\ll 1$ and $f''(0)\ge \mathrm{Const.\,}>0,$ then $f(s)$ is a strictly convex and non-negative function so there exists
a unique point $s_0$ close to $0$, where $f(t)$ attains its minimum and we have $f(s)\ge (s-s_0)^2/C$. Integrating this inequality, we see
that $\int_0^tf(s)ds\ge t^3/C$ and (\ref{eq9.98}) follows.

\medskip
\noindent
It remains to treat the case when $0\le f(0)\ll 1$, $f''(0)\le \epsilon \ll 1$ and $f^{(4)}(0)\ge \mathrm{Const.\,}>0$. The function
$g(s)=\frac{1}{2}(f(s)+f(-s))$ is even and has the Taylor expansion
$$
g(s)=f(0)+\frac{1}{2}f''(0)s^2+\frac{1}{4!}f^{(4)}(0)s^4+{\cal O}(s^6),
$$
and is therefore a smooth and strictly convex function of $s^2$, non-negative for $s^2\ge 0$. Denote this function by $k(s^2)$ and
restrict the attention to the interval $[0,t^2]$. If $k'(0)\ge 0$, the strict convexity implies that $k(\tau )\ge \tau ^2/C$ on $[0,t^2]$. If
$k'(t^2)\le 0$ we have $k(\tau )\ge (\tau -t^2) ^2/C$. In the remaining case when $k'(0)\le 0$ and $k'(t^2)\ge 0$ there exists $\tau
_0\in [0,t^2]$ such that $k(\tau )\ge (\tau -\tau _0) ^2/C$ on $[0,t^2]$ and this is finally the conclusion in all three cases.

\medskip
\noindent
Thus, $g(s)\ge \frac{1}{C}(s^2-\tau _0)^2$ for $-t\le s\le t$ and by an easy calculation, we obtain that
$$
\int_{-t}^t f(s)ds=\int_{-t}^tg(s)ds\ge
\frac{1}{C}\int_{-t}^t(s^2-\tau _0)^2ds\ge \frac{t^5}{\widetilde{C}},\ \
0\le t\ll 1.
$$
In other words, $(\varphi_0 \circ \exp t\nu -\varphi_0 \circ \exp (-t\nu))(x) $
has a lower bound as in (\ref{eq9.98}), and we are allowed to vary the
point $x$ in a small neighborhood, so we get the same conclusion for
$(\varphi_0\circ \exp 2t\nu -\varphi_0)(x)$ and after replacing $2t$ by $t$, for $(\varphi_0 \circ \exp t\nu -\varphi_0)(x)$.
\end{proof}

\medskip
\noindent
From the statement of Theorem \ref{thm_main}, let us recall the set $\Omega _+\subset {\bf R}^{6n}$, defined as the connected component of the set
$$
\varphi_0 ^{-1}((-\infty ,\varphi_0(\widetilde{S}_0))),
$$
which contains $\widetilde{M}_+$. It follows from the property 1) prior to (\ref{eq9.97}) that
$$
\exp (t\nu ) (\Omega _+)\subset \Omega _+,\quad t \le 0,
$$
and Proposition \ref{prop1} gives that $\exp(t\nu )(K)$ converges to $\{\widetilde{M}_+\}$ when $t\to -\infty$, for every $K\Subset \Omega _+$.

\medskip
\noindent
The stable manifold theorem tells us that in a suitable neighborhood of $\widetilde{S}_0$ there is a unique curve $\Gamma $ (manifold of dimension 1),
which is stable under the $\nu $--flow in the sense that if $x\in \Gamma $ then $\exp (t\nu )(x)$ converges to $\widetilde{S}_0$, exponentially fast, when
$t\to +\infty $. The set $\Gamma _+:=\Gamma \cap \Omega _+$ is also invariant under the forward $\nu $--flow and is
the image of a (connected) curve. For $w_0\in \Gamma _+$, let us put $\gamma(t)=\exp (t\nu )(w_0)$, $t\in {\bf R}$. Then $\gamma $ is a smooth
curve in $\Omega _+$ and we have,
\begeq
\gamma (t)\to \begin{cases}\widetilde{M}_+,\ &t\to -\infty ,\\
    \widetilde{S}_0,\ &t\to +\infty .
\end{cases}
\endeq
The trajectory $\gamma(t)$ will play a crucial role in the proof of Theorem \ref{thm_main}.

\bigskip
\noindent
Let us now resume the analysis of the eikonal equation (\ref{eq9}). The perturbation $\delta W(x)$ will be chosen so that $M_+\notin {\rm supp}\,(\delta W)$,
and since we are interested in smooth solutions $\varphi$ of (\ref{eq8.2}), for which $\varphi(\widetilde{M}_+) = \varphi'(\widetilde{M}_+) = 0$,
$\varphi''(\widetilde{M}_+)>0$, as we saw above, we have $\varphi = \varphi_0$ in a \neigh{} of $\widetilde{M}_+=(M_+,0,M_+)$. We shall therefore study
the solvability of the problem (\ref{eq9}), assuming that $\psi=0$ in a \neigh{} of $(M_+,0,M_+)$.

\begin{prop}
\label{prop2}
Let $\delta W\in C^{\infty}(\real^{2n})$ be such that $M_+\notin {\rm supp}\,(\delta W)$ and assume that $\delta W(x_1,x_2)$ is a
homogeneous polynomial of degree $m \geq 3$ in $x_2$. If $\psi\in C^{\infty}(\Omega)$, for $\Omega_+\Subset \Omega$, satisfies {\rm (\ref{eq9})} with
$\psi=0$ near $(M_+,0,M_+)$, then we have in $\Omega_+$,
$$
\psi(x,y,z) = {\cal O}((x_2,y_2,z_2)^2).
$$
\end{prop}
\begin{proof}
We shall view (\ref{eq9}) as a Hamilton-Jacobi equation of the form,
$$
p(x,y,z,\psi'_{x,y,z})=0,
$$
where
\begin{multline}
\label{eq9.981}
p(x,y,z,\xi,\eta,\zeta) = y\cdot \xi + \gamma(z-x)\cdot\zeta - \left(\partial_x W_0(x)+x-z\right)\cdot\eta \\
+ \frac{\gamma}{2}\sum_{j=1}^2 \alpha_j \zeta_j^2
-\partial_x \delta W\cdot \eta -2\partial_x \delta W \cdot \partial_y \varphi_0.
\end{multline}
In what follows, it will be convenient to write $w=(w_1,w_2)$, $w_j=(x_j,y_j,z_j)$, and $\omega=(\omega_1,\omega_2)$,
$\omega_j=(\xi_j,\eta_j,\zeta_j)$, for $j=1,2$. We know that the Lagrangian manifold $\Lambda_{\psi} = \{(w,\psi'(w))\}$ is the $H_p$--flowout
of the set
$$
\Lambda_{\psi}\cap {\rm neigh}((M_+,0,M_+;0,0,0),T^* \real^{6n}).
$$
To be precise, let
$\rho(0) = (w(0);\omega(0))\in {\rm neigh}((M_+,0,M_+;0,0,0), T^* \real^{6n})$
be such that $\omega(0)=0$, so that $\rho(0)\in \Lambda_{\psi} \subset p^{-1}(0)$, and consider the corresponding Hamiltonian trajectory
$$
(w(t);\omega(t)) = \rho(t) = \exp(tH_p)(\rho_0) \in \Lambda_{\psi}.
$$
We shall be interested in the trajectories $\rho(t)$ for which $w_2(0)=0$. It follows from
(\ref{eq2}) and (\ref{eq9.981}) that
$$
\partial_{w_2} p(w,\omega) = {\cal O}((w_2,\omega_2)),\quad \partial_{\omega_2} p(w,\omega) = {\cal O}((w_2,\omega_2)),
$$
and hence the Hamilton equations
$$
\dot{w_2}(t) = \partial_{\omega_2} p(w(t),\omega(t)),\quad \dot{\omega_2}(t) = -\partial_{w_2}p(w(t),\omega(t)),
$$
imply that
$$
\left(\dot{w_2}(t),\dot{\omega_2}(t)\right) = {\cal O}((w_2(t),\omega_2(t))).
$$
Along a trajectory, for which $w_2(0) = \omega_2(0)=0$, we have therefore $w_2(t)\equiv 0$, $\omega_2(t)\equiv 0$. A straightforward computation shows
next that along the set where $w_2=0$, $\omega_2=0$, we have
$$
\partial_{w_1} p = {\cal O}(\omega_1),
$$
while
$$
\partial_{\omega_1}p\cdot \partial_{w_1} = \nu_1(w_1,\partial_{w_1}) + {\cal O}(\omega_1)\cdot \partial_{w_1},
$$
where $\nu_1(w_1,\partial_{w_1})= \gamma(z_1-x_1)\cdot \partial_{z_1} + y_1\cdot \partial_{x_1} - \left(\partial_{x_1} W_1(x_1)+x_1-z_1\right)\cdot \partial_{y_1}$.
It follows that an $H_p$--trajectory $\rho(t) = (w_1(t),w_2(t);\omega_1(t),\omega_2(t))$ for which $w_2(0)=0$, $\omega(0)=0$, satisfies
$$
\rho(t) = (\exp(t\nu_1)(w_1(0)),0;0,0).
$$
From the definition of $\Lambda_\psi$ we know that
\begin{equation} \label{xidx}
d\psi = \omega dw,
\end{equation}
 so that
$\psi'(w_1(t), w_2(t)) = (\omega_1(t),\omega_2(t)) = (0,0)$. Thus, $\psi'(w_1,0) = 0$ for all $w_1$ such that $(w_1, 0) \in \Omega_+$.
Now from (\ref{xidx}), by a classical formula given, for instance, in Chapter 1 of~\cite{DiSj}, we also have
\begeq
\label{eq9.982}
\psi(w(t)) = \psi(w(0)) + \int_0^t \omega(s)\cdot \partial_{\omega} p (w(s),\omega(s))\,ds,
\endeq
where we know that $\psi(w(0))=0$, so that $\psi(w_1, 0) = 0$ again for
all $w_1$ such that $(w_1, 0) \in \Omega_+$. Using this and $\psi'(w_1,0) = 0$, we get that $\psi(w) = {\mathcal O}(w_2^2)$ and the proof is complete.
\end{proof}

\bigskip
\noindent
In what follows, we shall assume that $\delta W\in C^{\infty}(\real^n)$ has the properties described in Proposition \ref{prop2}. Using also that
$$
\partial_{y_j} \varphi_0 = \frac{1}{\alpha_j} y_j,\quad j=1,2,
$$
we see that the right hand side of (\ref{eq9}) is homogeneous of degree $m$ in $w_2 = (x_2,y_2,z_2)$. If the problem (\ref{eq9}) has a smooth solution
$\psi\in C^{\infty}(\Omega)$, then we have a Taylor expansion at $w_2=0$, writing also $w_1 = (x_1,y_1,z_1)$,
\begeq
\label{eq10}
\psi(w_1,w_2) \simeq \sum_{k=0}^{\infty} \psi_k(w_1,w_2).
\endeq
Here $\psi_k(w_1,w_2)$ is homogeneous of degree $k$ in $w_2$, with $C^{\infty}$--coefficients in $w_1$, and by Proposition \ref{prop2}, we know that
$\psi_0$ and $\psi_1$ vanish in $\Omega_+$. It follows from (\ref{eq9.5}) that $\nu \psi_k$ is also homogeneous of degree $k$ in $w_2$. We see then
that the term, homogeneous of degree $\mu \geq 0$, in the left hand side of (\ref{eq9}), is given by
\begin{multline}
\label{eq11}
\nu \psi_{\mu} + \frac{\gamma}{2}\alpha_1 \sum_{k=0}^{\mu} (\partial_{z_1} \psi_k) (\partial_{z_1} \psi_{\mu - k}) +
\frac{\gamma}{2}\alpha_2 \sum_{k=0}^{\mu} (\partial_{z_2} \psi_{k+1}) (\partial_{z_2} \psi_{\mu-k+1}) \\
-\partial_{x_1} \delta W \cdot \partial_{y_1} \psi_{\mu - m} - \partial_{x_2} \delta W \cdot \partial_{y_2} \psi_{\mu+2-m}.
\end{multline}
Here it is understood that $\psi_j \equiv 0$, for $j<0$.

\bigskip
\noindent
It is now easy to conclude that $\psi_{\mu}$ all vanish, for $\mu <m$, on the open set $\Omega_+\subset \real^{6n}$.
%
Indeed, taking $\mu=2$ in (\ref{eq11}), we see that the sum
$$
\sum_{k=0}^2 (\partial_{z_1} \psi_k) (\partial_{z_1} \psi_{2-k})
$$
vanishes in $\Omega_+$, while the only non-vanishing term in the sum
$$
\sum_{k=0}^2 (\partial_{z_2} \psi_{k+1}) (\partial_{z_2} \psi_{2-k+1})
$$
is given by $(\partial_{z_2} \psi_2)^2$. We get the equation
\begeq
\label{eq14}
\nu \psi_2 + \frac{\gamma}{2} \alpha_2 (\partial_{z_2} \psi_2)^2 = 0,
\endeq
where we also know that $\psi_2$ vanishes near $(M_+,0,M_+)$. The preceding equation can be viewed as a first order ordinary differential equation for $\psi_2$ along the integral curves of $\nu$ in $\Omega_+ \cap \{ (w_1, w_2) \ ; \ w_2=0 \}$ and we can
conclude that $\psi_2 = 0$ in $\Omega_+$.

\medskip
\noindent
We shall now see, arguing inductively, that $\psi_{\mu}=0$ in $\Omega_+$, for $\mu < m$. Indeed, assume that
$\psi_0 = \psi_1 = \psi_2 = \ldots\, = \psi_{\mu-1}=0$, for some $2<\mu < m$. Then we have
$$
\sum_{k=0}^{\mu} (\partial_{z_1} \psi_k) (\partial_{z_1} \psi_{\mu-k}) = 0.
$$
As for the sum
$$
\sum_{k=0}^{\mu} (\partial_{z_2} \psi_{k+1}) (\partial_{z_2} \psi_{\mu-k+1}),
$$
we see that the only term here which is not clearly vanishing corresponds to the case when $k+1 = \mu$. In this case, the corresponding term is equal to
$(\partial_{z_2}\psi_{\mu}) \partial_{z_2} \psi_2$, which vanishes after all. The sum above consequently vanishes, and from (\ref{eq11}) we get the equation
$$
\nu \psi_{\mu}=0,\quad \mu < m.
$$
It follows that $\psi_{\mu}=0$ in $\Omega_+$, for $\mu <m$. We conclude that a smooth solution $\psi$ of (\ref{eq9}), such that $\psi=0$ near $(M_+,0,M_+)$,
has the following form in $\Omega_+$,
$$
\psi(w_1,w_2) = \psi_m(w_1,w_2) + {\cal O}(w_2^{m+1}).
$$
Combining (\ref{eq9}) and (\ref{eq11}), we see that $\psi_m$ should satisfy the following non-homogeneous transport equation in $\Omega_+$,
\begeq
\label{eq15}
\nu \psi_m = 2\partial_x \delta W \cdot \partial_y \varphi_0.
\endeq

\bigskip
\noindent
The proof of Theorem \ref{thm_main} will therefore be concluded, once we establish the following result.

\begin{prop}
\label{prop3}
There exist a positive definite quadratic form $W_2(x_2)$, a Morse function $W_1(x_1)$ with two local minima and a saddle point, and a perturbation
$\delta W\in C^{\infty}(\real^{2n})$ with $M_+\notin {\rm supp}\left(\delta W\right)$, $S_0 \notin {\rm supp}\left(\delta W\right)$,
and such that $\delta W(x_1,x_2)$ is a homogeneous polynomial of degree $m\geq 3$ in $x_2$, for which the transport equation {\rm (\ref{eq15})}
does not have a smooth solution in $\Omega_+\cup \{\widetilde{S}_0\}$.
\end{prop}
\begin{proof}
With the notation $w_j=(x_j,y_j,z_j)$, $j=1,2$, let us write $\nu = \sum_{j=1}^2 \nu _j(w_j,\partial _{w_j})$. The preceding discussion shows
that there exists an integral curve $\gamma _1$ of $\nu _1$ such that
$$
\gamma_1(t)\to \begin{cases}(m_+,0,m_+) &t\to -\infty ,\\(s_0,0,s_0) &t\to +\infty .\end{cases}
$$

\medskip
\noindent
Let $N_2$ be the coefficient matrix of the linear vector field $\nu_2$, which we shall view as the linearization at $w_2=0$, and let
$\lambda _1,...,\lambda _{3n}$ be the corresponding eigenvalues, so that $\Re \lambda _j>0$, $1\leq j \leq 3n$. Let us assume, as we may, that $N_2$ has
no Jordan blocks, and after a linear change of variables, we may therefore assume that $w_2=(\omega _1,...,\omega _{3n})$, and
$$
\nu _2=\sum_{j=1}^{3n}\lambda_j \omega _j\partial _{\omega _j}.
$$
Then, writing $\lambda = (\lambda_1,\ldots\, \lambda_{3n})\in \comp^{3n}$, we have
$$
\nu _2(\omega ^\alpha )= \left({\lambda }\cdot \alpha\right) \omega^\alpha.
$$

\medskip
\noindent
Let us consider the equation (\ref{eq15}),
\begeq
\label{eq16}
(\nu _1+\nu _2)(\psi_m)=\frac{2}{\alpha _1}y_1\cdot \partial_{x_1}\delta W+\frac{2}{\alpha _2}y_2\cdot \partial_{x_2}\delta W,
\endeq
and put
$$
\psi_m = \frac{2}{\alpha _1}\delta W(x)+u.
$$
Then we get
\begeq
\label{eq17}
(\nu _1+\nu _2)(u)=\left(\frac{2}{\alpha_2}-\frac{2}{\alpha _1}\right)y_2\cdot \partial _{x_2}\delta W,
\endeq
where we recall that $\alpha _2\ne \alpha _1$. Using that $\delta W(x_1,x_2)$ is a homogeneous polynomial of degree $m\geq 3$ in $x_2$ for every $x_1$,
we may write
$$
\left(\frac{2}{\alpha _2}-\frac{2}{\alpha _1}\right)y_2\cdot \partial _{x_2}\delta
W=\sum_{|\alpha |=m}g_\alpha (x_1)\omega ^\alpha,
$$
and if (\ref{eq17}) has a smooth solution $u$, we can assume without
loss of generality that
$$
u=\sum_{|\alpha |=m}u_\alpha (w_1)\omega ^\alpha .
$$
The equation (\ref{eq17}) reduces then to the following decoupled system of equations,
\begeq
\label{eq18}
\left(\nu _1(w_1,\partial _{w_1})+{\lambda }\cdot \alpha\right) u_\alpha (w_1)=g_\alpha (x_1),\quad  \abs{\alpha} = m.
\endeq

\medskip
\noindent
We shall choose $\delta W(x_1,x_2) = \pm\psi (x_1)v(x_2)$, where $0\le \psi \in C_0^\infty ({\bf R}^n)$, $m_+,s_0\not\in \mathrm{supp\,}\psi$, while
$\psi >0$ somewhere on the image of $\gamma _1$, and where $v$ is a homogeneous polynomial of degree $m$. Then for some $\alpha =\alpha _0$ of length
$m$, we have $0\le g_\alpha \in C_0^\infty ({\bf R}^n)$, $m_+,s_0\not\in \mathrm{supp\,}g_\alpha $, and $g_\alpha $ is $>0$
somewhere on the image of $\gamma _1$.

\medskip
\noindent
The equation (\ref{eq18}) along $\gamma_1$ with $\alpha=\alpha _0$ has a compactly supported right hand side with constant
sign, not identically equal to 0, so the solution $\widetilde{u}=u_{\alpha _0}$ is non-zero either on $\gamma _1\cap \mathrm{neigh\,}(m_+,0,m_+)$ or on
$\gamma _1\cap \mathrm{neigh\,}(s_0,0,s_0)$.

\medskip
\noindent
In the first case we use the natural parametrization $\gamma_1(t)=\exp t\nu _1(w_0)$, $-\infty <t<+\infty $, where $w_0$ is some
fixed point on $\gamma _1$. Then there is a constant $C>0$ such that $\mathrm{dist\,}(\gamma _1(t),(m_+,0,m_+))\le C e^{-|t|/C}$ for $t\le 0$. For $t$ large
and negative, we have
$$
\left(\frac{d}{dt}+{\lambda }\cdot \alpha \right)u_\alpha (\gamma _1(t))=0,
$$
and hence,
$$
u_\alpha (\gamma (t))=Ce^{-({\lambda }\cdot \alpha) t },\ C\ne 0.
$$
Thus $u_\alpha $ is unbounded near $(m_+,0,m_+)$, and hence cannot be smooth near that point.

\bigskip
\noindent
In the second case, we recall that $\gamma _1$ is a part of the one-dimensional stable manifold through $(s_0,0,s_0)$ for the
$\nu _1$--flow. We can find new smooth local coordinates $x=(x_1,x'')$ centered at $(s_0,0,s_0)$, such that this stable manifold is given by $x''=0$, and hence
$$
\nu _1=a_1(x)\partial _{x_1}+\sum_{j=2}^{3n} a_j(x)\partial _{x_j},
$$
where $a_j(x_1,0)=0$ when $j\ge 2$. Furthermore, $a_1(x_1,0)=-\mu_1(x_1+f(x_1))$ where $\mu _1>0$, $f(x_1)={\cal O}(x_1^2)$ and we may
assume that $\gamma _1$ coincides with the positive $x_1$--axis near $(s_0,0,s_0)$. Along $\gamma _1$ and near $(s_0,0,s_0)$ we know that $u_\alpha $ is a
non-vanishing solution of the following equation,
$$
(-\mu _1(x_1+f(x_1))\partial _{x_1} + {\lambda}\cdot \alpha )u_\alpha =0,\ 0<x_1\ll 1,
$$
so that
$$
u_\alpha (x_1,0)=C\,\exp \left(\frac{{\lambda }\cdot \alpha }{\mu _1}\int_{x_1^0}^{x_1}\frac{1}{s+f(s)}ds \right),
$$
where $C\ne 0$ and $x_1^0>0$ is small and fixed. Here
$$
\frac{1}{s+f(s) }=\frac{1}{s}\,\frac{1}{1+\frac{f(s)}{s}}=\frac{1}{s}-\frac{f(s)}{s^2}+\frac{f(s)^2}{s^3}....,
$$
so
$$
\int_{x_1^0}^{x_1}\frac{1}{s+f(s)}ds =\ln x_1 +g(x_1)
$$
where $g$ is smooth near $x_1=0$. Thus,
$$
u_\alpha (x_1,0)=C\,x_1^{\frac{{\lambda }\cdot \alpha}{\mu _1}}e^{g(x_1)},\ C\ne 0,
$$
and if $\frac{{\lambda }\cdot \alpha}{\mu _1}\not\in {\bf N}$ (which can be arranged by choosing the parameters suitably, c.f. (\ref{eq9.8})),
we conclude that $u_\alpha$ cannot be smooth near $(s_0,0,s_0)$. The proof of Proposition \ref{prop3} is complete.
\end{proof}


\end{document}